\documentclass[runningheads,a4paper]{llncs}
\usepackage[left=3cm,right=3cm,top=3cm,bottom=3cm]{geometry}
\usepackage{amssymb}
\usepackage{graphicx}
\usepackage{amsmath, amsfonts}
\usepackage{stmaryrd}
\usepackage{mathrsfs}

\usepackage{tikz}
\usetikzlibrary{arrows}

\usepackage{url}
\urldef{\mail}\path|agnieszka.kulacka11@imperial.ac.uk|    
\newcommand{\keywords}[1]{\par\addvspace\baselineskip
\noindent\keywordname\enspace\ignorespaces#1}

\renewcommand{\restriction}{\mathord{\upharpoonright}}

\makeatletter
\newcommand*{\@old@slash}{}\let\@old@slash\slash
\def\slash{\relax\ifmmode\delimiter"502F30E\mathopen{}\else\@old@slash\fi}
\makeatother

\spnewtheorem{defn}[cc]{Definition}{\bfseries}{\upshape}
\spnewtheorem{thm}[cc]{Theorem}{\bfseries}{\upshape}
\spnewtheorem{lem}[cc]{Lemma}{\bfseries}{\upshape}
\spnewtheorem{expl}[cc]{Example}{\bfseries}{\upshape}
\spnewtheorem{cor}[cc]{Corollary}{\bfseries}{\upshape}
\spnewtheorem{rem}[cc]{Remark}{\bfseries}{\upshape}

\begin{document}

\mainmatter

\title{A tableau for set-satisfiability for extended fuzzy logic  BL}
\author{Agnieszka Ku\l acka\thanks{Supported by EPSRC DTA studentship in Computing Department, Imperial College London. I would like to thank Professor Ian Hodkinson for numerous readings of this paper and his constructive comments. }}
\institute{Department of Computing, Imperial College London\\
\mail\\
\url{http://www.doc.ic.ac.uk/~ak5911/}}

\authorrunning{A tableau for set-satisfiability for extended fuzzy logic  BL}
\maketitle

\begin{abstract}
This paper presents a tableau calculus for finding a model for a set-satisfiable finite set of formulas of a fuzzy logic BL$_{\triangle\sim}$, a fuzzy logic BL with additional  Baaz connective $\triangle$ and the involutive negation $\sim$, if such a model exists. The calculus is a generalisation of a tableau calculus for BL, which is based on the decomposition theorem for a continuous t-norm. The aforementioned tableau calculus for BL is used to prove that a formula $\psi$ of BL is valid with respect to all continuous t-norms or to find a continuous t-norm $\star$ and assignment $V$ of propositional atoms to [0,1] such that $\star$-evaluation $V_\star(\psi)<1$. The tableau calculus presented in this paper enables for a finite set of formulas $\Psi$  of BL$_{\triangle\sim}$ and $\mathcal{K}\subseteq[0,1]$ to find a continuous t-norm $\star$ and assignment $V$ of propositional atoms to [0,1] such that $\star$-evaluation $V_\star(\psi)\in \mathcal{K}$ for all $\psi\in \Psi$, or alternatively to show that such a model does not exist.

\end{abstract}

\section{Introduction}
Classical satisfiability of a propositional formula in a structure (model) is understood as the truth value of the formula relative to an assignment of truth values to propositional atoms. We say that a formula is satisfiable if such an assignment exists, in which it is true. In mathematical fuzzy logic a formula has a truth value, which is a value from [0,1], where 0 is the absolute falsity, 1 is the absolute truth, and other values indicate a partial truth. We can also consider a wider notion of satisfiability than classical. Following \cite{Butnariu1995} and \cite{Navara2000}, we can work with $\mathcal{K}$-satisfiability, where $\mathcal{K}\subseteq [0,1]$. A formula $\psi$ is $\mathcal{K}$-satisfiable if its truth value belongs to $\mathcal{K}$ under some assignment of values from [0,1] to atoms. To be able to define set-satisfiability rigorously, we need an exposition of fuzzy logic BL$_{\triangle\sim}$.

Basic fuzzy logic BL is a propositional fuzzy logic, in which formulas of BL are written with propositional atoms, $\bar{0}$ (falsum) and $\bar{1}$ (verum), joined by $\&$ (strong conjuncton), $\rightarrow,\vee, \wedge$ (weak conjunction),  $\leftrightarrow$. (see \cite{Hajek1998}) Logic BL$_{\triangle\sim}$ is BL with additional unary connectives Baaz connective $\triangle$ and the involutive negation $\sim$. The semantics of the logic is defined as follows. A t-norm $\star$ (also called a residuated t-norm if it has a residuum) is a function defined on  $[0,1]^2$ with values in $[0,1]$ such as it is associative, commutative, non-decreasing with neutral element 1, and its residuum (if it exists) is a function $\Rightarrow:[0,1]^2\rightarrow[0,1]$ satisfying $x\star z\leq y$ iff $z\leq x\Rightarrow y$.  We define the function $\Delta:[0,1]\rightarrow[0,1]$ by $\Delta 1=1$ and $\Delta x=0$ for all $x\in [0,1)$. Given  BL$_{\triangle\sim}$ formulas $\psi, \varphi$, the assignment $V$ of propositional atoms to elements of [0,1], a residuated t-norm $\star$, we inductively define $\star$-evaluation $V_\star$ as $V_\star(p)=V(p)$ for atoms $p$, $V_\star(\bar{0})=0$,  $V_\star(\bar{1})=1$, $V_\star(\psi\&\varphi)=V_\star(\psi)\star V_\star(\varphi)$, $V_\star(\psi\rightarrow\varphi)=V_\star(\psi)\Rightarrow V_\star(\varphi)$, $V_\star(\psi\vee\varphi)=\max\{V_\star(\psi), V_\star(\varphi)\}$, $V_\star(\psi\wedge\varphi)=\min\{V_\star(\psi), V_\star(\varphi)\}$, $V_\star(\triangle\psi)=\Delta V_\star(\psi)$, $V_\star(\sim\psi)=1-V_\star(\psi)$. Our work on satisfiability will be based on a continuous residuated t-norm.

We can now formally define $\mathcal{K}$-satisfiability. A formula $\psi$ is $\mathcal{K}$-satisfiable if there exists a continuous residuated t-norm $\star$ and an assignment $V$ of atoms to [0,1], for which  $V_\star(\psi)\in \mathcal{K}$. In this paper we will build a model for a finite set of formulas of BL$_{\triangle\sim}$ that is $\mathcal{K}$-satisfiable in this model. To achieve this, we will use tableau methods as a semantic proof system. The idea is based on decomposition theorem (see \cite{Cintula2011},  \cite{Hajek1998}, \cite{Metcalfe2009}, \cite{MostertShields1957}), by which a continuous t-norm (and thus its residuum) is expressed as a family of Product and \L ukasiewicz components. The formulas, for which we try to show that are $\mathcal{K}$-satisfiable, are translated into tableau formulas and the complement of the set $\mathcal{K}$ with respect to [0,1] is expressed as a union of subintervals. The latter enables to create finite branches of the tableau, i.e. finite sequences of tableau formulas, in which the next element of the sequence has fewer symbols of interpreted connectives. The branches are extended using some rules that we defined. 

The main result is that the tableau for a finite set of formulas for a subset $\mathcal{K}$ of truth values [0,1] is open iff the  set of formulas is $\mathcal{K}$-satisfiable. The tableau calculus presented in this paper enables for $\mathcal{K}\subseteq[0,1]$ and a finite set of formulas $\Psi$  of BL$_{\triangle\sim}$ to find a continuous t-norm $\star$ and an assignment $V$ of propositional atoms to [0,1] such that $V_\star(\psi)\in \mathcal{K}$ for all $\psi\in \Psi$, or alternatively to show that such a model does not exist. 

 The axiomatization of BL$_{\triangle\sim}$ has been shown to be finitely strong standard complete  with respect to every continuous residuated t-norm with the additional unary connectives $\triangle,\sim$ (see \cite{Flaminio}). If $\mathcal{K}=[0,1)$ and our tableau closes, that is there is no such model for which the finite set of formulas are not  tautologies, the set of formulas are provable in BL$_{\triangle\sim}$. 

It is worth noticing that there exist tableau calculi or other proof systems for BL or \L ukasiewicz logics that demonstrate that a formula is a tautology or that it is not (see \cite{Bova}, \cite{Kulacka2014},  \cite{Kulacka2013}, \cite{Montagna}, \cite{Olivetti2003}, \cite{OrlowskaGolinska2011}, \cite{Vetterlein}, \cite{Vidal}). The advantage of our tableau over the ones existing in the literature is three-fold, (1) in case that the formula is not a tautology, it constructs a countermodel, (2) we can show that a formula is $\mathcal{K}$-satisfiable for $\mathcal{K}\subseteq[0,1]$, (3) our calculus tackles a set of formulas from an extended logic BL with additional unary connectives.

The paper is organised in the following way: section 2 recalls the decomposition theorem for continuous t-norm (see \cite{Cintula2011},  \cite{Hajek1998}, \cite{Metcalfe2009}, \cite{MostertShields1957}), section 3 defines the $\mathcal{K}$-tableau for a set of formulas of BL$_{\triangle\sim}$, which construction is exemplified in section 4. Section 5 shows the main result, i.e. the equivalence of $\mathcal{K}$-satisfiability  of a finite set of formulas of BL$_{\triangle\sim}$ and the existence of an open branch in $\mathcal{K}$-tableau. Conclusions are presented in the final section.

\section{Decomposition theorem}
In this section we will recall the decomposition theorem since we will use it in the proofs for a generalised tableau calculus for a set of formulas of BL$_{\triangle \sim}$ that are set-satisfiable. In this theorem we employ a special case of an ordinal sum of a family of continuous components.
\begin{defn}
Let $0\leq a<b\leq 1$. A \textit{\L ukasiewicz component} is a function $\star_{a,b}:[a,b]^2\rightarrow [a,b]$ such that for every $x,y\in [a,b]$, $x\star_{a,b} y=\max\{a,x+y-b\}$.
\end{defn}

\begin{defn}
 Let $0\leq a<b\leq 1$. A \textit{Product component} is a function $\star_{a,b}:[a,b]^2\rightarrow [a,b]$ such that for every $x,y\in [a,b]$, $x\star_{a,b} y=a+\frac{(x-a)(y-a)}{b-a}$.
\end{defn}
\begin{defn}\label{def1}
Let $C$ be a countable index set and let $([a_n^{K_n},b_n^{K_n}])_{n\in C}$ be a family of closed intervals with $0\leq a_n^{K_n}<b_n^{K_n}\leq 1, K_n\in\{L, P\}$ such that their interiors are pairwise disjoint. 
An \textit{ordinal sum of the family of continuous components}, $(\star_{a_n^{K_n}, b_n^{K_n}})_{n\in C}$, is the function $\star: [0,1]^2\rightarrow[0,1]$ such that for every $x,y\in [0,1]$
\begin{equation}\label{eq3}
x\star y=\begin{cases}
x\star_{a_n^{K_n}, b_n^{K_n}} y&\text{if } x,y\in [a_n^{K_n},b_n^{K_n}] \\
\min\{x,y\}&\text{otherwise,}\\
\end{cases}
\end{equation}
where $\star_{a_n^{K_n}, b_n^{K_n}}$ is either \L ukasiewicz component if $K_n=L$ or Product component if $K_n=P$.
\end{defn}
\begin{thm}\label{thm1}(Decomposition theorem \cite{Cintula2011},  \cite{Hajek1998}, \cite{Metcalfe2009}, \cite{MostertShields1957})\\
The function $\star:[0,1]^2\rightarrow[0,1]$ is a continuous t-norm iff
it is an ordinal sum of a family of  continuous components $(\star_{a_n^{K_n}, b_n^{K_n}})_{n\in C}$, where $C$ is a countable set and $([a_n^{K_n},b_n^{K_n}])_{n\in C}$ is a family of closed intervals, and the open intervals $(a_n^{K_n},b_n^{K_n})$, $n\in C, K_n\in\{L,P\}$,  are pairwise disjoint.
\end{thm}

\noindent From theorem \ref{thm1}, it follows that the residuum of a continuous t-norm  of $\star$ is the function $\Rightarrow: [0,1]^2\rightarrow[0,1]$ given by equation (\ref{eq4}). For every $x,y\in [0,1]$,
\begin{equation}\label{eq4}
x\Rightarrow y=\begin{cases}
1&\text{if } x\leq y,\\
b_i^L-x+y&\text{if } x>y \text{ and } x,y\in[a_i^L,b_i^L], \\
a_i^P+\frac{(y-a_i^P)(b_i^P-a_i^P)}{x-a_i^P}&\text{if } x>y \text{ and } x,y\in[a_i^P,b_i^P],\\
y,&\text{otherwise.}\\
\end{cases}
\end{equation}

\noindent We will use these functions given in (\ref{eq3}) and (\ref{eq4}) in building the tableau calculus in the next section.

\section{$\mathcal{K}$-tableau for a set of formulas}
In \cite{Kulacka2014} we defined a tableau calculus and proved its soundness and completeness with respect to continuous t-norms. This enabled us to demonstrate that a given formula of BL is either valid or we could find a model, in which its truth value is less than 1. The generalisation presented below works in two ways: (1) we will have a tableau calculus for a finite set of formulas of BL$_{\triangle\sim}$, (2) we will have a calculus to show that they are $\mathcal{K}$-satisfiable for $\mathcal{K}\subseteq [0,1]$.

 To be able to accommodate set-satisfiability within a tableau calculus we will express a subset of [0,1] as a union of subintervals of [0,1].

\begin{lem}\label{tab_20}
Any subset $\mathcal{K}$ of [0,1] can be expressed as a union of
pairwise disjoint maximal subintervals of $\mathcal{K}$.
\end{lem}
\begin{proof}
We will inductively build the union. \textit{General case.} Suppose that we have built a set of intervals whose union we call $\hat{\mathcal{K}}\subseteq \mathcal{K}$. Take any $k_0\in \mathcal{K}-\hat{\mathcal{K}}$. Let $k^{-}=\inf\{k: k\leq k_0 \wedge \forall t: k<t\leq k_0\rightarrow t\in \mathcal{K}\}$ and $k^{+}=\sup\{k: k_0\leq k \wedge \forall t: k_0\leq t<k\rightarrow t\in \mathcal{K}\}$. If $k^{-}, k^{+}\in \mathcal{K}$, then we have a closed interval $[k^{-}, k^{+}]$, if $k^{-}\in K, k^{+}\not\in \mathcal{K}$, then a right-open interval $[k^{-}, k^{+})$, if $k^{-}\not\in \mathcal{K}, k^{+}\in \mathcal{K}$, then a left-open interval $(k^{-}, k^{+}]$, and if $k^{-}, k^{+}\not\in \mathcal{K}$, then we have an open interval $(k^{-}, k^{+})$. We add the interval with the endpoints $k^-,k^+$ to the union $\hat{\mathcal{K}}$. This concludes the general case. When $\mathcal{K}-\hat{\mathcal{K}}=\emptyset$, the process terminates and we have selected a union of pairwise disjoint  maximal subintervals, which is obviously equal to $\mathcal{K}$.
\qed
\end{proof}

We will now fix $\mathcal{K}$, and express its complement $\mathcal{K}'=[0,1]-\mathcal{K}$ as a union of pairwise disjoint maximal subintervals $\{J_i:i\in I\}$ of $\mathcal{K}'$. The left endpoint of each $J_i$ is denoted by $j_i^-$ and the right endpoint of $J_i$ is denoted by $j_i^+$.

We will recall the definition of a tableau formula, a translation function, some notions of graph theory that are necessary for defining $\mathcal{K}$-tableau of a set of formulas. We modified these definitions to suit our purpose.

\begin{defn}(Tableau formula)\label{tab_1}\\
Let $I$ be an index set, and $Const=\{c_i^-, c_i^+: i\in I\}$, where $c_i^-, c_i^+$ are pairwise distinct constants. Let $L_0=Par{ }\cup{ }\{ + , -, \cdot, \div, \min, \max, \leq, < \}\cup Const$ and $L_1=L_0\cup\{\star, \Rightarrow,\Delta\}$ be signatures, where $Par$ is a set of constants (parameters), $+ , -, \cdot, \div, \min, \max, \star$, $\Rightarrow$ are binary function symbols, $\Delta$ is a unary function symbol and $\leq, <$ are binary relation symbols. Let $Var$ be a set of variables.
\begin{enumerate}
\item Let $x$ be an $L_1$-term. A \textit{disjunct formula} $\eta_{J_i}(x)$ is a formula saying $x\not\in J_i$ and defined as\\($x${ }$\Diamond^{-}_{J_i}${ }$c^{-}_i)\vee (c^{+}_i${ }$\Diamond^{+}_{J_i}${ } $x$), where
\begin{enumerate}
\item $\Diamond^{-}_{J_i}$ is $<$ and $\Diamond^{+}_{J_i}$ is $<$ if ${J_i}=[j_i^{-}, j_i^{+}]$,
\item $\Diamond^{-}_{J_i}$ is $\leq$ and $\Diamond^{+}_{J_i}$ is $<$ if ${J_i}=(j_i^{-}, j_i^{+}]$,
\item $\Diamond^{-}_{J_i}$ is $<$ and $\Diamond^{+}_{J_i}$ is $\leq$ if ${J_i}=[j_i^{-}, j_i^{+})$, 
\item $\Diamond^{-}_{J_i}$ is $\leq$ and $\Diamond^{+}_{J_i}$ is $\leq$ if ${J_i}=(j_i^{-}, j_i^{+})$.
\end{enumerate}
\item If $x, y$ are $L_1$-terms, then $x\leq y, x<y, x=y, \eta_{J_i}(x)$ are \textit{tableau formulas}. If $x,y$ are $L_0$-terms, then $x\leq y, x<y, x=y$ are $L_0$-\textit{formulas}.
\item An $L_0$-structure $\mathcal{M}$ is called \textit{standard} iff it is of the form $$(\mathbb{R},  + , -, \cdot, \div, \min, \max, \leq, <,  0, 1, \rho, (j_i^-, j^+_i) :i\in I),$$ where $+ , -, \cdot, \div$, $\min, \max, 0, 1, \leq, <$ are the usual functions with $x\div 0$ assigned to $0$ for any $x\in \mathbb{R}$, $\rho:Par\rightarrow [0,1]$ is a function, and, $c_i^-,c_i^+$ are interpreted in $\mathcal{M}$ as $j_i^-, j_i^+$, respectively.
\item  Let $E$ be a set of tableau formulas $e$ of the form $s\leq t, s<t, s=t$, where $s,t$ are $L_0$-terms. We say that a mapping $$\sigma:Var\rightarrow [0,1]$$ is a \textit{solution} of $E$ iff there exists a standard $L_0$-structure $\mathcal{M}$ such that $$\mathcal{M}, \sigma\models e,$$ for all $e\in E$. We will call $\mathcal{M}$ an \textit{$L_0$-structure modelling} $E$. 
\end{enumerate}
\end{defn}

\noindent We will use $PROP$ to denote a set of propositional atoms.

\begin{defn}(Translation function)\label{tab_2}\\
Let $\mathcal{F}$ be the set of formulas of BL$_{\triangle\sim}$ and $\mathbb{T}$ be the set of $L_1$-terms.  Let $\mu: PROP\rightarrow Var$ (we will write $\mu(p)$ as $\mu_p$) be a one-to-one mapping assigning variables to propositional atoms. Let $\psi, \varphi\in \mathcal{F}$. Then, we define a translation function $\tau:\mathcal{F}\rightarrow\mathbb{T}$, inductively:
\begin{enumerate}
\item $\tau(\bar{0})=0$, $\tau(\bar{1})=1$,
\item $\tau(p)=\mu_p$ for every $p\in PROP$,
\item $\tau(\psi\&\varphi)=\tau(\psi)\star\tau(\varphi)$,
\item $\tau(\psi\rightarrow\varphi)=\tau(\psi)\Rightarrow\tau(\varphi)$,
\item $\tau(\psi\vee\varphi)=\max\{\tau(\psi),\tau(\varphi)\}$,
\item $\tau(\psi\wedge\varphi)=\min\{\tau(\psi),\tau(\varphi)\}$,
\item $\tau(\triangle \psi)=\Delta \tau(\psi)$,
\item $\tau(\sim\psi)=1-\tau(\psi)$.
\end{enumerate}
\end{defn}

We will recall definitions of some terms of graph theory. A \textit{graph} is a structure $(N, E)$, where $N$ is a set of nodes, $E$ is a set of edges such that $E\subseteq N\times N$ such that $\lnot E(n,n)$ for all $n\in N$. A \textit{successor} of $n\in N$ is $n'\in N$ iff there is an edge $e$ such that $(n,n')=e$. A \textit{predecessor} of $n\in N$ is $n'\in N$ iff there is an edge $e'$ such that $(n',n)=e'$. A \textit{path} from $n\in N$ is a sequence of nodes $n_0=n, n_1, ..., n_k, ...$ such that $k\geq 0$ and $n_i$ is a predecessor of $n_{i+1}$ for $i=0, 1, ..., k, ...$. A \textit{leaf} is a node with no successors, and a \textit{root} is a node with no predecessors. A \textit{branch} is either a path from a root to a leaf if the latter exists, or otherwise an infinite path from a root. We will call a \textit{tree} an acyclic connected graph $(N, E)$, in which there is exactly one root and if a node is not a root, then it has exactly one predecessor. The depth of a node $n\in N$ within a branch $\mathcal{B}$, denoted by $d(n, \mathcal{B})$, is the number of nodes on the path from the root to $n$. The height of a node $n\in N$ within a branch $\mathcal{B}$, denoted by $h(n, \mathcal{B})$, is the number of nodes on the path from $n$ to a leaf of branch $\mathcal{B}$ if $\mathcal{B}$ is finite.\\

The definition below is a much extended version of the tableau calculus in \cite{Kulacka2014}, it incorporates additional connectives, a set of formulas (as oposed to one formula as in \cite{Kulacka2014}) and the extended notion of satisfiability. The exposition of branch expansion rules is more compact than in \cite{Kulacka2014} with new rules for splitting and for the additional connectives. This compactness of the rules enables us for a finite set of formulas to have finite branches with possibly infinite nodes as the latter depends on set $\mathcal{K}$, and this was not necessary for the tableau calculus in \cite{Kulacka2014}. Finiteness of the branches of a $\mathcal{K}$-tableau is essential for the proof of theorem \ref{sound_2}.

\begin{defn}\label{tab_3}
Let $\Psi$ be a set of BL$_{\triangle\sim}$ formulas. A \textit{$\mathcal{K}$-tableau} $\mathcal{T}$ for $\Psi$  is a tree whose nodes are sets of tableau formulas and whose root is $$\{\eta_{J_i}(\tau(\psi)):i\in I, \psi\in\Psi\},$$ and on which the branch expansion rules\footnote{We use branch expansion rules to generate nodes in a branch.}  have been fully applied. Let $\Gamma$ be a set of tableau formulas.

First, we will well-order $\Psi$, and then we will be selecting $\psi\in \Psi$ one by one and applying Split Rule to all $\eta_{J_i}(\tau(\psi)), i\in I$ simultaneously in all current nodes. Note that there will be $2^{|I|}$ nodes generated for each $\psi\in \Psi$ by application of Split Rule. \\
\newpage

\noindent \textbf{Split Rule}. Let $S\subseteq I$. For each $S$, there is a successor of  $\Gamma\cup\{\eta_{J_i}(\tau(\psi)):i\in I\}$ given by:
\begin{description}
\item[$S_\psi$.] $\Gamma \cup  \{\tau(\psi)${ }$\Diamond_{J_i}^{-}${ }$c_i^{-}:i\in S\} \cup \{c_i^{+}${ }$\Diamond_{J_i}^{+}${ }$\tau(\psi): i\in I-S\}$
\end{description}
where $\Diamond_{J_i}^{-}, \Diamond_{J_i}^{+}$ are as defined in Definition \ref{tab_1}.\\

Next, when there is no more $\eta_{J_i}(\tau(\psi)), i\in I, \psi\in \Psi$ in the current nodes, we will apply the other branch expansion rules. \\

\noindent The multiple inequality should be understood in the usual way, e.g. instead of writing $a\leq c, c= d, d<f$, we write $a\leq c= d<f$. Let $x,y$ be $L_1$-terms.   Let $K, K_0,..., K_{n-1}\in\{L, P\}$ be the labels as shown in the branch expansion rules. Suppose that parameters $a_0^{K_0}< b_0^{K_0}\leq a_1^{K_1}<...\leq a_{n-1}^{K_{n-1}}< b_{n-1}^{K_{n-1}}$ ($n\geq 1$) have been selected in the previous steps. We will use the following sets $\mathcal{I}^K$ ($\mathcal{J}$ respectively) in the subrules \textbf{\L, P} (\textbf{min}, respectively) of the branch expansion rules $\star, \Rightarrow$, for which we have chosen the active term (one undergoing substitution) of the form $x\star y, x\Rightarrow y$.\\

\noindent In the following case, $a^K, b^K\in Par$ are new distinct parameters. Then 
\begin{itemize}
\item Case 1. $\mathcal{I}^K=\{0\leq a^K< b^K\leq a_0^{K_0}\}$,
\item Case 2. $\mathcal{I}^K=\{a^K=a_i^{K_i}< b^K=b_i^{K_i}\}$ for some $0\leq i\leq n-1$ such that $K=K_i$,
\item Case 3. $\mathcal{I}^K=\{b_i^{K_i}\leq a^K< b^K\leq a_{i+1}^{K_{i+1}}\}$ for some $0\leq i\leq n-2$,
\item Case 4. $\mathcal{I}^K=\{b_{n-1}^{K_{n-1}}\leq a^K< b^K\leq 1\}$.
\item Case 5. $\mathcal{J}=\{0\leq x\leq a_0^{K_0}\}$,
\item Case 6. $\mathcal{J}=\{a_i^{K_i}\leq x\leq b_i^{K_i}, y\leq a_i^{K_i}\}$ for some $0\leq i\leq n-1$,
\item Case 7. $\mathcal{J}=\{a_i^{K_i}\leq x\leq b_i^{K_i}, b_i^{K_i}\leq y\}$ for some $0\leq i\leq n-1$,
\item Case 8. $\mathcal{J}=\{b_i^{K_i}\leq x\leq a_{i+1}^{K_{i+1}}\}$ for some $0\leq i\leq n-2$,
\item Case 9. $\mathcal{J}=\{b_{n-1}^{K_{n-1}}\leq x \leq 1\}$.
\end{itemize}
If no parameters have been selected in the previous steps, then $\mathcal{I}^K=\{0\leq a^K< b^K\leq 1\}$ (Case 10.)  and $\mathcal{J}=\emptyset$ (Case 11.).\\

\noindent Note that cases 1-4 and 10 express the conditions when $x,y$ belong to the same \L ukasiewicz or Product component, while the remaining cases when they do not. Cases 5, 8, 9, 11 take care of the situation when $x$ does not belong to any of the components (being constructed by introduction of the parameters by branch expansion rules), while cases 6 and 7 are for situations, in which $x$ is in one of the components, and $y$ is not in this component. \\

\noindent We will use a notation $\gamma[v/t]$ to denote the result of substituting the term $v$ for each occurrence of the term $t$ (if any) in the formula $\gamma$.\\

\noindent If a node consists wholly of $L_0$-formulas, it is a leaf and no rules are applied to it.  Otherwise, we choose an active term $t$ of the form $x\star y, x\Rightarrow y$, or $\Delta x$ that occurs in at least one formula in the node, and apply the rule below according to the form of the active term. \\

\noindent \textbf{Rule $(\star)$}. A branch with a node $\Gamma$ expands following the subrules:
\begin{description}
\item[\L .] $\mathcal{I}^L \cup\{a^L\leq x\leq b^L, a^L\leq y\leq b^L\}\cup\{\gamma[\max\{a^L, x+y-b^L\}/x\star y]:\gamma \in \Gamma$\}
\item[P.] $\mathcal{I}^P \cup \{a^P\leq x\leq b^P, a^P\leq y\leq b^P\}\cup\{\gamma[a^P+\frac{(x-a^P)(y-a^P)}{b^P-a^P}/x\star y]: \gamma\in\Gamma$\}
\item[min.] $\mathcal{J} \cup \{\gamma[\min\{x,y\}/x\star y]: \gamma\in\Gamma\}$
\end{description}

\noindent \textbf{Rule $(\Rightarrow)$}. A branch with a node $\Gamma$ expands following the subrules:
\begin{description}
\item[All.] $\{x\leq y\}\cup\{\gamma[1/x\Rightarrow y]: \gamma\in\Gamma$\}
\item[\L .] $\mathcal{I}^L \cup\{a^L\leq y<x\leq b^L\}\cup\{\gamma[b^L-x+y/x\Rightarrow y]: \gamma\in\Gamma$\}
\item[P.] $\mathcal{I}^P \cup \{a^P\leq y< x\leq b^P\}\cup\{\gamma[a^P+\frac{(y-a^P)(b^P-a^P)}{x-a^P}/x\Rightarrow y]: \gamma\in\Gamma$\}
\item[min.] $\mathcal{J} \cup \{ y<x\}\cup\{[y/x\Rightarrow y]: \gamma\in\Gamma \}$
\end{description}
\newpage
\noindent\textbf{Rule $(\Delta)$}. A branch with a node $\Gamma$ expands following the subrules:
\begin{description}
\item[$\Delta$1.] $ \{1\leq x\}\cup\{\gamma[1/\triangle x]: \gamma\in\Gamma$\}
\item[$\Delta$2.] $\{x<1\}\cup\{\gamma[0/\triangle x]: \gamma\in\Gamma$\}
\end{description}

Note that the actual number of new nodes generated by rules $\star$ and $\Rightarrow$ will depend on how many parameters are on the current node as these influence the number of different $\mathcal{I}^L, \mathcal{I}^P, \mathcal{J}$. That is it dependes on how many cases of the subrules we can apply. For example, if there are four parameters on the current node $0\leq a_0^L<b_0^L\leq a_1^L<b_1^L\leq 1$, subrules \textbf{\L} and \textbf{P} will have 5 (Cases 1, 2, 2, 3, 4) and 3 (Cases 1, 3, 4) cases, respectively, and a subrule \textbf{min} will have 7 cases (Cases 5, 6, 6, 7, 7, 8, 9). So we get a total of 15 successors for the current node generated by \textbf{Rule} $\star$ and 16 successors generated by \textbf{Rule} $\Rightarrow$. 

This concludes Definition \ref{tab_3}.
\end{defn}

\begin{defn}\label{tab_4}
For each  branch $\mathcal{B}$ of a $\mathcal{K}$-tableau $\mathcal{T}$ and each node $n\in\mathcal{B}$, we consider the set of $L_0$-formulas in $n$, $n\restriction_{L_0}$. We say that $\mathcal{B}$ is \textit{closed} if for some node $n\in\mathcal{B}$, $n\restriction_{L_0}$ has no solution, otherwise it is \textit{open}. A $\mathcal{K}$-tableau $\mathcal{T}$ is \textit{closed} if it only contains closed branches. \footnote{Whether a tableau is closed is decidable by Tarski theorem \cite{Tarski} on decidability of the first-order theory of $(\mathbb{R},+,\cdot)$.} A $\mathcal{K}$-tableau $\mathcal{T}$ is \textit{open} if it has an open branch.
\end{defn}

\section{Example}
Now we will show an example to demonstrate the applicability of the rules.  We will not show a whole $\mathcal{K}$-tableau for $\Psi$, just an open branch.

\begin{expl}\label{example1} Let $\mathcal{K}=[\frac{1}{2},\frac{3}{4}]\cup\{1\}$ and $\Psi=\{ \bar{1}\rightarrow p\& r, \triangle r\rightarrow (p\vee q)\}$. Let us build a $\mathcal{K}$-tableau for $\Psi$. $J_1=[0,\frac{1}{2}), J_2=(\frac{3}{4},1)$. Note that nodes (11), (12), (13), (14) are generated simultaneously. Let $\psi_1=\bar{1}\rightarrow p\& r$ and $\psi_2=\triangle r\rightarrow (p\vee q)$.
\begin{description}
\item[(1)] $\{1\Rightarrow(\mu_p\star\mu_r)< 0 \vee \frac{1}{2}\leq 1\Rightarrow(\mu_p\star\mu_r),$\\
$ 1\Rightarrow(\mu_p\star\mu_r)\leq \frac{3}{4} \vee 1\leq1\Rightarrow(\mu_p\star\mu_r), $\\
$\Delta \mu_r\Rightarrow\max\{\mu_p,\mu_q\}< 0 \vee \frac{1}{2}\leq \Delta \mu_r\Rightarrow\max\{\mu_p,\mu_q\}, $\\
$\Delta \mu_r\Rightarrow\max\{\mu_p,\mu_q\}\leq \frac{3}{4} \vee 1\leq\Delta \mu_r\Rightarrow\max\{\mu_p,\mu_q\}
 \}$
\item[(11)] Split Rule. $S_{\psi_1}=\{1,2\}$\\
$\{1\Rightarrow(\mu_p\star\mu_r)< 0, 1\Rightarrow(\mu_p\star\mu_r)\leq \frac{3}{4},$\\ $\Delta \mu_r\Rightarrow\max\{\mu_p,\mu_q\}< 0 \vee \frac{1}{2}\leq \Delta \mu_r\Rightarrow\max\{\mu_p,\mu_q\}, $\\
$\Delta \mu_r\Rightarrow\max\{\mu_p,\mu_q\}\leq \frac{3}{4} \vee 1\leq\Delta \mu_r\Rightarrow\max\{\mu_p,\mu_q\}
 \}$
\item[(12)] Split Rule. $S_{\psi_1}=\{2\}$\\
$\{\frac{1}{2}\leq 1\Rightarrow(\mu_p\star\mu_r), 1\Rightarrow(\mu_p\star\mu_r)\leq \frac{3}{4},$\\
$ \Delta \mu_r\Rightarrow\max\{\mu_p,\mu_q\}< 0 \vee \frac{1}{2}\leq \Delta \mu_r\Rightarrow\max\{\mu_p,\mu_q\}, $\\
$\Delta \mu_r\Rightarrow\max\{\mu_p,\mu_q\}\leq \frac{3}{4} \vee 1\leq\Delta \mu_r\Rightarrow\max\{\mu_p,\mu_q\}
 \}$
\item[(13)] Split Rule. $S_{\psi_1}=\{1\}$\\
$\{1\Rightarrow(\mu_p\star\mu_r)< 0, 1\leq1\Rightarrow(\mu_p\star\mu_r),$\\
$ \Delta \mu_r\Rightarrow\max\{\mu_p,\mu_q\}< 0 \vee \frac{1}{2}\leq \Delta \mu_r\Rightarrow\max\{\mu_p,\mu_q\}, $\\
$\Delta \mu_r\Rightarrow\max\{\mu_p,\mu_q\}\leq \frac{3}{4} \vee 1\leq\Delta \mu_r\Rightarrow\max\{\mu_p,\mu_q\}
 \}$
\item[(14)] Split Rule. $S_{\psi_1}=\emptyset$.\\
$\{\frac{1}{2}\leq1\Rightarrow(\mu_p\star\mu_r), 1\leq1\Rightarrow(\mu_p\star\mu_r),$\\ $\Delta \mu_r\Rightarrow\max\{\mu_p,\mu_q\}< 0 \vee \frac{1}{2}\leq \Delta \mu_r\Rightarrow\max\{\mu_p,\mu_q\}, $\\
$\Delta \mu_r\Rightarrow\max\{\mu_p,\mu_q\}\leq \frac{3}{4} \vee 1\leq\Delta \mu_r\Rightarrow\max\{\mu_p,\mu_q\}
 \}$
\item[(124)]  Split Rule. $S_{\psi_2}=\emptyset$.\\
$\{\frac{1}{2}\leq 1\Rightarrow(\mu_p\star\mu_r), 1\Rightarrow(\mu_p\star\mu_r)\leq \frac{3}{4},$\\
$\frac{1}{2}\leq \Delta \mu_r\Rightarrow\max\{\mu_p,\mu_q\}, 1\leq \Delta \mu_r\Rightarrow\max\{\mu_p,\mu_q\}
 \}$
\item[(1241)] Rule ($\Rightarrow$) All.\\
$\{\frac{1}{2}\leq \Delta \mu_r\Rightarrow\max\{\mu_p,\mu_q\}, 1\leq \Delta \mu_r\Rightarrow\max\{\mu_p,\mu_q\},$\\
$1\leq \mu_p\star\mu_r, \frac{1}{2}\leq 1, 1\leq \frac{3}{4}\}$ \\ (closed)
\item[(1242)] Rule ($\Rightarrow$) \L . Case 10.\\
$\{\frac{1}{2}\leq \Delta \mu_r\Rightarrow\max\{\mu_p,\mu_q\}, 1\leq \Delta \mu_r\Rightarrow\max\{\mu_p,\mu_q\},$\\
$0\leq a_0^L<b_0^L\leq 1, a_0^L\leq \mu_p\star\mu_r<1\leq b_0^L, \frac{1}{2}\leq b_0^L-1+\mu_p\star\mu_r, b_0^L-1+\mu_p\star\mu_r\leq \frac{3}{4}\}$ 
\item[(12421)] Rule ($\star$) \L . Case 2.\\
$\{\frac{1}{2}\leq \Delta \mu_r\Rightarrow\max\{\mu_p,\mu_q\}, 1\leq \Delta \mu_r\Rightarrow\max\{\mu_p,\mu_q\}, 0\leq a_0^L<b_0^L\leq 1,$\\
$a_0^L= a_1^L<b_1^L=b_0^L, a_1^L\leq \mu_p\leq b_1^L, a_1^L\leq \mu_r\leq b_1^L,$\\ $a_0^L\leq \max\{a_1^L, \mu_p+\mu_r-b_1^L\}<1\leq b_0^L, \frac{1}{2}\leq b_0^L-1+ \max\{a_1^L, \mu_p+\mu_r-b_1^L\},$\\
$ b_0^L-1+ \max\{a_1^L, \mu_p+\mu_r-b_1^L\}\leq \frac{3}{4}\}$ 
\item[(124211)] Rule ($\Delta$) $\Delta 1$.\\
$\{ 0\leq a_0^L<b_0^L\leq 1,a_0^L= a_1^L<b_1^L=b_0^L, a_1^L\leq \mu_p\leq b_1^L, a_1^L\leq \mu_r\leq b_1^L,$\\
$a_0^L\leq \max\{a_1^L, \mu_p+\mu_r-b_1^L\}<1\leq b_0^L, \frac{1}{2}\leq b_0^L-1+ \max\{a_1^L, \mu_p+\mu_r-b_1^L\},$\\
$ b_0^L-1+ \max\{a_1^L, \mu_p+\mu_r-b_1^L\}\leq \frac{3}{4},$\\
$ 1\leq \mu_r, \frac{1}{2}\leq 1\Rightarrow\max\{\mu_p,\mu_q\}, 1\leq 1\Rightarrow\max\{\mu_p,\mu_q\}\}$

\item[(1242111)] Rule ($\Rightarrow$) All.\\
$\{ 0\leq a_0^L<b_0^L\leq 1,a_0^L= a_1^L<b_1^L=b_0^L, a_1^L\leq \mu_p\leq b_1^L, a_1^L\leq \mu_r\leq b_1^L, $\\
$a_0^L\leq \max\{a_1^L, \mu_p+\mu_r-b_1^L\}<1\leq b_0^L, \frac{1}{2}\leq b_0^L-1+ \max\{a_1^L, \mu_p+\mu_r-b_1^L\},$\\
$ b_0^L-1+ \max\{a_1^L, \mu_p+\mu_r-b_1^L\}\leq \frac{3}{4}, 1\leq \mu_r,$\\$ 1\leq\max\{\mu_p,\mu_q\}, \frac{1}{2}\leq 1, 1\leq 1\}$

\end{description}
No more branch explansion rules can be applied to node (1242111). Let $$\{(\mu_p,\frac{1}{2}), (\mu_q,1), (\mu_r,1)\}\subseteq\sigma$$ and $$\mathcal{M}=\biggl(\mathbb{R}, + , -, \cdot, \div, \leq, <, 0,1, \rho, (0, \frac{1}{2}), (\frac{3}{4}, 1)\biggr),$$ where $\{(a_0^L,0),$ $(b_0^L,1), (a_1^L,0)$, 
$(b_1^L,1)\}\subseteq \rho$. Then $$\mathcal{M}, \sigma\models e,$$ for all $e\in \{ 0\leq a_0^L<b_0^L\leq 1,a_0^L= a_1^L<b_1^L=b_0^L, a_1^L\leq \mu_p\leq b_1^L, a_1^L\leq \mu_r\leq b_1^L, a_0^L\leq \max\{a_1^L, \mu_p+\mu_r-b_1^L\}<1\leq b_0^L, \frac{1}{2}\leq b_0^L-1+ \max\{a_1^L, \mu_p+\mu_r-b_1^L\}, b_0^L-1+ \max\{a_1^L, \mu_p+\mu_r-b_1^L\}\leq \frac{3}{4}, 1\leq \mu_r, 1\leq\max\{\mu_p,\mu_q\}, \frac{1}{2}\leq 1, 1\leq 1\}$. Therefore the branch with node (1242111) is open, so the $\mathcal{K}$-tableau for $\Psi$ is open. Note that this is not the only possible solution and it is not the only open branch.
\end{expl}
It is worth noticing that branches (11), (13) will close in any model.

\section{Finite $\mathcal{K}$-satisfiability}
In this section we will show that every finite set $\Psi$ of formulas of BL$_{\triangle\sim}$ is  $\mathcal{K}$-satisfiable for an arbitrary $\mathcal{K}\subseteq [0,1]$ iff we can construct an open $\mathcal{K}$-tableau for $\Psi$. Since $\Psi$ is finite, the $\mathcal{K}$-tableaux we will work with in this section all have finite branches.

\begin{defn}
Let $\Psi$ be a  set of formulas of BL$_{\triangle\sim}$ and $\mathcal{K}\subseteq [0,1]$. We say that $\Psi$ is \textit{$\mathcal{K}$-satisfiable} if there exists $V:PROP\rightarrow [0,1]$ and a continuous t-norm $\star$ such that $V_\star(\psi)\in \mathcal{K}$ for all $\psi\in \Psi$. Formula $\psi$ is $\mathcal{K}$-satisfiable if the singleton $\{\psi\}$ is $\mathcal{K}$-satisfiable.
\end{defn}

\begin{expl}
Let $\mathcal{K}=[\frac{1}{2},\frac{3}{4}]\cup\{1\}$. The set $\Psi=\{ \bar{1}\rightarrow p\& r, \triangle r\rightarrow (p\vee q)\}$ is $\mathcal{K}$-satisfiable. Take $V(p)=\frac{1}{2}, V(q)=1, V(r)=1$ and \L ukasiewicz t-norm as $\star$. Then $V_\star(\bar{1}\rightarrow p\& r)=1\Rightarrow (V(p)\star V(r))=1\Rightarrow (V(p)+ V(r)-1)=1\Rightarrow (\frac{1}{2}+1-1)=1\Rightarrow \frac{1}{2}=1-1+\frac{1}{2}=\frac{1}{2}$, $V_\star(\triangle r\rightarrow (p\vee q))=\Delta V(r)\Rightarrow\max\{V(p),V(q)\}=1\Rightarrow\max\{\frac{1}{2}, 1\}=1\Rightarrow 1=1$, and thus $V_\star\psi)\in \mathcal{K}$ for all $\psi\in \Psi$. We use the model constructed in example \ref{example1} to find values $V(p), V(q), V(r)$ and an ordinal sum, which consists of \L ukasiewicz component [0,1].
\end{expl}

\begin{defn}\label{sound_1}
 Let $\{a_k^{K_k}, b_k^{K_k}:k\in C\}$, where $0\leq a_k^{K_k}< b_k^{K_k}\leq 1$ and $K_k\in\{L, P\}$, be the parameters introduced by branch extension rules in $\mathcal{B}$. Let $l_\mathcal{B}$ be the leaf of the branch. Suppose that $\mathcal{M}$ is a standard $L_0$-structure and $\sigma: Var\rightarrow [0,1]$ such that 
\begin{equation}\notag
\mathcal{M}, \sigma \models s,
\end{equation}
for all $s\in l_\mathcal{B}$.
\begin{enumerate} 
\item We expand the model $\mathcal{M}$ to a standard $L_1$-structure $\mathcal{M}_\mathcal{B}=
(\mathbb{R},  + , -, \cdot, \div, \min, \max, \leq, <,$ $0, 1, \rho, (j_i^-, j_i^+): i\in I, {\star_\mathcal{B}}, {\Rightarrow}_\mathcal{B}, \Delta)$, where $j_i^-, j_i^+$ are endpoints of interval $J_i$, such that for every $v,w\in[0,1]$,\\
\footnotesize
$v\star_\mathcal{B} w=
\begin{cases}
\max\{\rho(a_k^L), v+w-\rho(b_k^L)\}&\text{if } v, w\in [\rho(a_k^L),\rho(b_k^L)],\\
\rho(a_k^P)+\frac{(v-\rho(a_k^P))\cdot(w-\rho(a_k^P))}{\rho(b_k^P)-\rho(a_k^P)}&\text{if } v,w\in [\rho(a_k^P),\rho(b_k^P)], \\
\min(v,w)&\text{otherwise,}\\
\end{cases}$\\

$v\Rightarrow_\mathcal{B} w=
\begin{cases}
1&\text{if }v\leq w,\\
\rho(b_k^L)-v+w&\text{if } \rho(a_k^L)\leq w<v \leq\rho(b_k^L), 0\leq k\leq n-1,\\
\rho(a_k^P)+\frac{(w-\rho(a_k^P))\cdot(\rho(b_k^P)-\rho(a_k^P))}{v-\rho(a_k^P)}&\text{if } \rho(a_k^P)\leq w<v \leq\rho(b_k^P), 0\leq k\leq n-1,\\
w&\text{otherwise.}\\
\end{cases}$\\

$\Delta v=
\begin{cases}
1&\text{if }1\leq v,\\
0&\text{if }v<1.\\
\end{cases}$\\

\normalsize

\item $\mathcal{M}_\mathcal{B}, \sigma\models z_1\vee z_2$ iff $\mathcal{M}_\mathcal{B}, \sigma\models z_1$ or $\mathcal{M}_\mathcal{B}, \sigma\models z_2$.
\item A subset  $S$  of a node of branch $\mathcal{B}$ is $\mathcal{B}$-\textit{satisfiable via $\mathcal{M}, \sigma$} iff for all $s\in S$ $$\mathcal{M}_\mathcal{B}, \sigma\models s,$$ 
where $L_1$-structure $\mathcal{M}_\mathcal{B}$ is constructed from $\mathcal{M}$ as in 1.
\end{enumerate}
\end{defn}

\begin{thm}\label{sound_2}
Let $\mathcal{K}\subseteq [0,1]$. Let $\Psi$ be a finite set of BL$_{\triangle\sim}$ formulas and $\mathcal{T}$ be a $\mathcal{K}$-tableau, whose root is $\{\eta_i(\tau(\psi)):i\in I, \psi\in\Psi\}$ constructed as in Definition \ref{tab_3}. Then the following are equivalent:
\begin{enumerate}
\item $\mathcal{T}$ has an open branch.
\item $\Psi$ is $\mathcal{K}$-satisfiable.
\end{enumerate}
\end{thm}

\begin{proof}
There are finitely many formulas in $\Psi$, therefore each branch of $\mathcal{T}$ is finite.\footnote{Though they may be infinitely many branches with infinite nodes.} Suppose that branch $\mathcal{B}$ of tableau $\mathcal{T}$ is open.  Since it is finite, a leaf $l_\mathcal{B}$ of branch $\mathcal{B}$ exists. Also, for every node $n$ of branch $\mathcal{B}$, $n\restriction_{L_0}\subseteq l_\mathcal{B}$. Since $\mathcal{B}$ is open, there is a standard $L_0$-structure $\mathcal{M}$ modelling the leaf $l_\mathcal{B}$ of $\mathcal{B}$ and an assignment $\sigma:Var\rightarrow[0,1]$ such that $\mathcal{M},\sigma\models l_\mathcal{B}$. We will  construct model $\mathcal{A}=([0,1],\star,\Rightarrow,0,1,V)$ such that $V_\star(\psi)\in \mathcal{K}$ for every $\psi\in\Psi$. First, we put $V(p)=\sigma(\tau(p))$ for all $p\in PROP$.  We define operation $\star:[0,1]^2\rightarrow [0,1]$ as $\star_\mathcal{B}$ and operation $\Rightarrow:[0,1]^2\rightarrow [0,1]$ as $\Rightarrow_\mathcal{B}$ (see definition \ref{sound_1}). By theorem \ref{thm1} and its  corollary expressed as formula (\ref{eq4}), $\star$ is a continuous t-norm with residuum $\Rightarrow$ and also by definitions \ref{tab_2} and \ref{sound_1}, $V_\star(\psi)=\llbracket\tau(\psi)\rrbracket^{\mathcal{M}_\mathcal{B},\sigma}$, where $\llbracket z\rrbracket^{\mathcal{M}_\mathcal{B},\sigma}$ is the value of $L_1$-term $z$ in $\mathcal{M}_\mathcal{B}$ under the assignment $\sigma$. By induction on $h(m_\mathcal{B}, \mathcal{B})$, we show the claim that every node $m_\mathcal{B}$ of $\mathcal{B}$ is $\mathcal{B}$-satisfiable via $\mathcal{M},\sigma$. The sketch of the proof is as follows. Let $m'_\mathcal{B}$ be the node of $\mathcal{B}$ such that $h(m'_\mathcal{B},\mathcal{B})=h(m_\mathcal{B},\mathcal{B})+1$,  where $m_\mathcal{B}$ is assumed to be $\mathcal{B}$-satisfiable via $\mathcal{M},\sigma$. Then $m'_\mathcal{B}$ is also $\mathcal{B}$-satisfiable via $\mathcal{M},\sigma$ by definition \ref{tab_3}.
Therefore, in particular by this claim, the root of $\mathcal{T}$ is $\mathcal{B}$-satisfiable via $\mathcal{M},\sigma$, and thus by definitions \ref{tab_2} and \ref{sound_1}, $\llbracket\tau(\psi)\rrbracket^{\mathcal{M}_\mathcal{B},\sigma}=V_\star(\psi)\in \mathcal{K}$ for every $\psi\in\Psi$. Therefore, $\Psi$ is $\mathcal{K}$-satisfiable.

Conversely, suppose that there is a model $\mathcal{A}=([0,1], {\star}, {\Rightarrow}, 0, 1, V)$, where $\star$ is a continuous t-norm and $\Rightarrow$ is its residuum, $V:PROP\rightarrow [0,1]$ such that $V_\star(\psi)\in\mathcal{K}$  for all $\psi\in\Psi$.  We know that for every node $n_\mathcal{B}$ of a branch $\mathcal{B}$, we have  $n_\mathcal{B}\restriction_{L_0}\subseteq l_\mathcal{B}$, where $l_\mathcal{B}$ is the leaf of the branch, which exists since all branches are finite. Therefore, a branch is open if its leaf $l_\mathcal{B}$ has a solution. That is, we need to find a standard $L_0$-structure modelling $l_\mathcal{B}$, say $\mathcal{M}$, and a mapping $\sigma:Var\rightarrow[0,1]$ such that $\mathcal{M}, \sigma\models e$ for all $e\in l_\mathcal{B}$. In (1) below, we construct the structure $\mathcal{M}$; that is we find the mapping $\rho:PAR\rightarrow[0,1]$. At the same time we choose nodes on a branch, say $\mathcal{B}$. Then, in (2) we show that branch $\mathcal{B}$ is open.

(1) Since we know the values $V_\star(\psi)$ for all $\psi\in\Psi$, we can select the nodes, which were generated by Split Rules applied to each $\eta_i(\tau(\psi)), i\in I$ for each $\psi\in\Psi$. Then we proceed in the following way. By theorem \ref{thm1}, $\star$ is defined as the ordinal sum of proto-t-norms $(\star_{\alpha_n^{K_n},\beta^{K_n}_n})_{n\in C}$, where $C$ is a countable index set. We will assign values of parameters occurring on $\mathcal{B}$ under $\rho$ to elements of $\{\alpha_n^{K_n},\beta^{K_n}_n:n\in C\}$. Suppose we selected the sequence of nodes $n_1,..., n_l$, where $l\geq 1$, $n_1$ is the root of $\mathcal{T}$ and $n_{i+1}$ is the successor of $n_i$ for all $1\leq i< l$. That is branch $\mathcal{B}$ is partially defined and $\rho$ is defined for all parameters on these nodes. By comparing $n_l$ and its successors, we deduce what the active term is and thus which of the branch expansion rules has been applied. Suppose it is $(\Rightarrow)$, the other cases are similar. Thus, there are subformulas of a formula in $\Psi$, $\psi_1, \psi_2$, and there is a set $\Gamma$ such that  $\tau(\psi_1)\Rightarrow\tau(\psi_2)$ is a subterm of all $\gamma\in\Gamma$, and $\Gamma$ is a subset of $n_l$, and none of $\gamma\in\Gamma$ belong to the successors of $n_l$. We know whether or not there is $i\in C$ such that $V_\star(\psi_1), V_\star(\psi_2)\in [\alpha_i^{K_i}, \beta_i^{K_i}]$, and whether or not $V_\star(\psi_1)\leq V_\star(\psi_2)$. If $V_\star(\psi_1)> V_\star(\psi_2)$ and if there is such $i$ and $K_i=L$ (or $K_i=P$), we select the subrule \L. of ($\Rightarrow$) (the subrule P. of ($\Rightarrow$), respectively), and depending on the relation of $\alpha_i^{K_i}, \beta_i^{K_i}$ to these $\alpha_j^{K_j}, \beta_j^{K_j}$ that are the values of parameters occurring on $n_l$, we select the node, say $n'_{l}$, resulting from Cases 1-4, and 10. Let $\rho$ assign to the parameters at $n'_{l}$ that do not occur in $n_{l}$, say $a,b$ with $a<b$, values $\rho(a)=\alpha_i^{K_i}, \rho(b)=\beta_i^{K_i}$. Suppose now there is not $i\in C$ such that $V_\star(\psi_1), V_\star(\psi_2)\in [\alpha_i^{K_i}, \beta_i^{K_i}]$ and $V_\star(\psi_1)>V_\star(\psi_2)$. We know the relations among $V_\star(\psi_1), V_\star(\psi_2)$, and the values of parameters occurring in $n_l$, thus we know which of Cases 5-9, and 11 match these values in the model $\mathcal{A}$. The remaining case is $V_\star(\psi_1)\leq V_\star(\psi_2)$, for which there is only one successor. Therefore, we can select the subsequent node.  We have now selected the next node in the path from the root. The procedure terminates at a leaf, where there are no $L_1$-terms occurring, at which point we selected all nodes in branch $\mathcal{B}$. We also partially defined the function $\rho$. To the parameters that have not received values under $\rho$ in this procedure, we assign arbitrary values from $[0,1]$. We have now constructed a standard $L_0$-structure $\mathcal{M}$.

(2)  To be able to show that $\mathcal{M}, \sigma\models e$ for all $e\in l_\mathcal{B}$, 
where $l_\mathcal{B}$ is the leaf of branch $\mathcal{B}$, it is sufficient to prove that for all nodes in $\mathcal{B}$, $n_\mathcal{B}$, $\mathcal{M}_\mathcal{B}, \sigma\models f$ for all $f\in n_\mathcal{B}$.  We will sketch the proof by induction on $d(n_\mathcal{B},\mathcal{B})$. We need to show that if $\mathcal{M}_\mathcal{B}, \sigma\models f$ for all $f\in n_\mathcal{B}$, then $\mathcal{M}_\mathcal{B}, \sigma\models f'$ for all $f'\in n'_\mathcal{B}$, where $d(n'_\mathcal{B}, \mathcal{B})=d(n_\mathcal{B}, \mathcal{B})+1$. By inspecting $n_\mathcal{B}, n'_\mathcal{B}$, we know which formulas are in $n_\mathcal{B}-n'_\mathcal{B}$. Suppose that in all formulas $\gamma\in n_\mathcal{B}-n'_\mathcal{B}$, $\tau(\psi_1)\Rightarrow\tau(\psi_2)$ is the active formula; the other cases are similar. By (1) above, we know which subrule of the branch expansion rule ($\Rightarrow$) and which of its cases are used to generate $n'_\mathcal{B}$. Thus, suppose that it was subrule \L{ }and Case 1 (again, the other cases are similar).  Let $a, b$ be the new parameters occurring at $n'_\mathcal{B}$. Thus, $\{0\leq a<b\leq a_0^L\}\subseteq n'_\mathcal{B}-n_\mathcal{B}$, where $a_0^L$ is a parameter occurring at $n_\mathcal{B}$. By (1) above, we know that $0\leq \rho(a)<\rho(b)\leq \rho(a_0^L)$. The other elements of $n'_\mathcal{B}-n_\mathcal{B}$ are (a) $a\leq \tau(\psi_2)<\tau(\psi_1)\leq b$, and (b) the elements that belong to  $\{\gamma[b-\tau(\psi_1)+\tau(\psi_2)/\tau(\psi_1)\Rightarrow\tau(\psi_2)]:\gamma\in n_\mathcal{B}-n'_\mathcal{B}\}$.

\noindent \textit{Claim}. Let $\theta$ be a subformula of a formula in $\Psi$. Then $\llbracket\tau(\theta)\rrbracket^{\mathcal{M}_\mathcal{B},\sigma}=V_\star(\theta)$.

\noindent \textit{Proof of the claim.} By induction on $\theta$. The base case for atomic $\theta$ is easy. Take subformulas of a formula in $\Psi$, $\theta, \varphi$ and assume the induction hypothesis for them. We show the result for $\theta\rightarrow\varphi$, i.e. $\llbracket\tau(\theta\rightarrow\varphi)\rrbracket^{\mathcal{M}_\mathcal{B},\sigma}=V_\star(\theta\rightarrow\varphi)$. The cases for $\theta\&\varphi, \theta\vee\varphi, \theta\wedge\varphi, \triangle\theta, \sim\theta$ are similar.

\noindent LHS=$\llbracket\tau(\theta\rightarrow\varphi)\rrbracket^{\mathcal{M}_\mathcal{B},\sigma}=\llbracket\tau(\theta)\Rightarrow\tau(\varphi)\rrbracket^{\mathcal{M}_\mathcal{B},\sigma}=\llbracket\tau(\theta)\rrbracket^{\mathcal{M}_\mathcal{B},\sigma}\Rightarrow_\mathcal{B}\llbracket\tau(\varphi)\rrbracket^{\mathcal{M}_\mathcal{B},\sigma}=V_\star(\theta)\Rightarrow_\mathcal{B} V_\star(\varphi)$, by inductive hypothesis.
 RHS=$V_\star(\theta\rightarrow\varphi)=V_\star(\theta)\Rightarrow V_\star(\varphi)$.
Now, by construction in (1), $\Rightarrow_\mathcal{B}, \Rightarrow$ agree on $V_\star(\theta), V_\star(\varphi)$ as long as $\theta, \varphi$ are subformulas of formulas in $\Psi$. Therefore, LHS=RHS. This completes the proof of the claim.

By (1) above and Claim, the inequalities in (a) and (b) are true in $\mathcal{M}_\mathcal{B}, \sigma$. Thus, $\mathcal{M}_\mathcal{B}, \sigma\models f'$ for all $f'\in n'_\mathcal{B}$. Therefore, we showed also that $\mathcal{M}, \sigma\models f'$ for all $f'\in n'_\mathcal{B}\restriction_{L_0}$. We have now proved that in particular $\mathcal{M}, \sigma\models e$ for all $e\in l_\mathcal{B}$. Thus, there is an $L_0$-structure modelling $l_\mathcal{B}$, $\mathcal{M}$, and a mapping $\sigma$ such that $\mathcal{M}, \sigma\models e$ for all $e\in l_\mathcal{B}$. Therefore $\mathcal{B}$ is open.
\qed
\end{proof}

\section{Conclusion and further research}
In this paper we constructed a $\mathcal{K}$-tableau calculus for formulas of BL$_{\triangle\sim}$ and proved that the existence of an open branch in the tableau is equivalent to $\mathcal{K}$-satisfiability of a finite set of formulas of BL$_{\triangle\sim}$. There are two additional notions of fuzzy satisfiability defined in \cite{Butnariu1995} and \cite{Navara2000}, which we will adapt to the context of BL$_{\triangle\sim}$ formulas.

\begin{defn}
Let $\Psi$ be a set of BL$_{\triangle\sim}$ formulas and $r\in[0,1]$. We say that $\Psi$ is \textit{weakly}, respectively \textit{strongly $r$-satisfiable} if $\Psi$ is $[r,1]$-satisfiable, respectively $\{r\}$-satisfiable. If there is  a maximal $r\in[0,1]$ for which $\Psi$ is weakly, respectively strongly satisfiable, then we call $r$ the \textit{weak}, respectively \textit{strong consistency degree} of $\Psi$. 
\end{defn}

Our result expressed in Definition \ref{tab_3} provides a method for constructing $[r,1]$-tableau, respectively $\{r\}$-tableau for a finite set of BL$_{\triangle\sim}$ formulas, which has finitely many branches. In this case, the tableau method due to Theorem \ref{sound_2} supplies an algorithm for constructing a continuous residuated t-norm and an evaluation for proposiional formulas if $[r,1]$-tableau, respectively $\{r\}$-tableau has an open branch, a problem that could be solved by an inequality solver. Given $r$, to check whether $\Psi$ is weakly $r$-satisfiable, one has to construct an adapted tableau with the following root $\{\sigma_r\leq \tau(\psi): \psi\in \Psi\}$, and a standard $L_0$-structure $\mathcal{M}$ modelling the leaf of a branch will be extended to $(\mathcal{M}, r)$, where $(\sigma_r)^{(\mathcal{M}, r)}=r$. Note that all the rules will be the same. Similarly, one can tackle strong $r$-satisfiability of finite $\Psi$. Having the above, one may envisage an implementation for the tableau methods such as described in \cite{Brys2012}.

\end{document}